\documentclass[runningheads, envcountsame, a4paper]{llncs}
\usepackage{llncsdoc}
\include{references.bib}

\usepackage{graphicx}
\usepackage{cite}
\usepackage{url}

\usepackage{amsmath}
\usepackage{amssymb}
\usepackage{amsthm}

\usepackage[usenames,dvipsnames]{color}
\usepackage{comment}

\includecomment{Long}

\newtheorem{conj}{Conjecture}
\usepackage{tikz}
\usetikzlibrary{arrows,automata, positioning}
\definecolor{light-gray}{gray}{0.8}

\providecommand{\com[2]}{\begin{tt}[#1: #2]\end{tt}}

\title{On Completely Reachable Automata\\
and Subset Reachability}
\titlerunning{Subset Reachability}

\author{Fran\c cois Gonze \and Rapha\"el M. Jungers \thanks{R. M. Jungers is a F.R.S.-FNRS Research Associate}
\thanks{This work was supported by the French Community of Belgium and by the IAP network DYSCO.}}

\institute{ICTEAM Institute\\
UCL, Louvain La Neuve, Belgium\\
 \email{$\{$francois.gonze,raphael.jungers$\}$@uclouvain.be}}

\begin{document}

\maketitle

\begin{abstract}{This article focuses on subset reachability in synchronizing automata. 
First, we provide families of synchronizing automata with subsets which cannot be reached with short words. These families do not fulfil Don's Conjecture about subset reachability. Moreover, they show that some subsets need exponentially long words to be reached, and that the restriction of the conjecture to included subsets also does not hold.
Second, we analyze completely reachable automata and provide a counterexample to the conjecture of Bondar and Volkov about the so-called $\Gamma_1$-graph. We finally prove an alternative version of this conjecture.}\end{abstract}

\section*{Introduction}\label{intro}

Automata\footnote{Formal definitions are provided in the next subsection.} are very useful tools in applied mathematics. In pattern recognition, they allow to parse texts and efficiently find letter sequences. In language theory, they are the basic tools defining formal languages and context free languages. In theoretical computer science, automata provide simple models for the behaviour of computing devices.
More recently, automata have also been at the core of synthesis and verification of complex automated systems. See \cite{berthe2010combinatorics, linz2011introduction} for references on the subject.

Synchronization is an important topic in automata theory. Indeed, if a machine can be modelled as a synchronizing automaton, then it is possible to fix its state by applying a sequence of commands corresponding to a synchronizing word. It has direct applications in robotics \cite{Moore56}, matrix theory \cite{JungersBlondelOlshevsky14}, consensus theory \cite{PYChev} and group theory \cite{ArCaSt15} among others. The most famous problem in this field was proposed by Jan \v Cern{\'y} in 1964 \cite{cernyPirickaRosenauerova64}. It states that if an automaton has a synchronizing word, then it also has a short synchronizing word: 

\begin{conj}[\v Cern{\'y}'s conjecture, 1964 \cite{cernyPirickaRosenauerova64}] \label{cernyconj}
Let $A= (Q;\Sigma;\delta)$ be a synchronizing automaton with $|Q|=n$.  Then, it has a synchronizing word of length at most $(n-1)^2.$\end{conj}

In  \cite{cerny64}, \v Cern\'y proposes an infinite family of automata attaining this bound, for any number of states. Fig.~\ref{C4} represents the automaton of this family with four states. 
Although many improvements have been achieved recently, both for some particular classes of automata (see \cite{steinberg-2009, BBP11, Dubuc98, kari03, GonzeGGJV17, trahtman_cerny}) or by improving the best general bound \cite{Szykula2017}, Conjecture~\ref{cernyconj} is not proven yet in its general formulation. A survey on the subject was provided by Mikhail Volkov \cite{volkov_survey}.
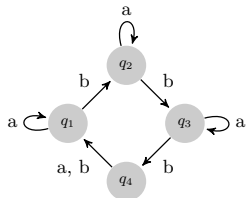
\begin{figure}[h!]
\begin{center}
\scalebox{0.8}{
\begin{tikzpicture}[->,>=stealth',shorten >=1pt,auto,node distance=1.7cm,
                    semithick]
  \tikzstyle{every state}=[fill=light-gray,draw=none,text=black, scale=0.8]

  \node[state] (A)                    {$q_1$};
  \node[state]         (B) [above right of=A] {$q_2$};
  \node[state]         (D) [below right of=A] {$q_4$};
  \node[state]         (C) [below right of=B] {$q_3$};

  \path (A) edge              node {b} (B)
            edge [loop left]  node {a} (C)
        (B) edge [loop above] node {a} (B)
            edge              node {b} (C)
        (C) edge              node {b} (D)
         	edge [loop right] node {a} (D)
        (D) edge  			  node {a, b} (A);
\end{tikzpicture}
}
\end{center}

\caption{A synchronizing automaton attaining \v Cern\'y's bound
}
\label{C4}
\end{figure}

Conjecture~\ref{cernyconj} can be adressed by the two natural and complementary approaches that are subset synchronization and subset reachability. In the subset synchronization approach, one considers the length of a shortest word which synchronizes a set $S$. It is studied in \cite{gonze2016synchronizing, DBLP:journals/corr/Vorel14, kirnasov2003glueing} among others. In the subset reachability approach, one considers the length of a shortest word reaching a set $S$. It is studied in \cite{don2016vcerny, bondar2016completely}. Subset reachability can be seen as a direct application of the extension method\footnote{In the extension method, one proves that any subset can be extended. Then, starting with a set $S$ and extending it recursively, one obtains the set $Q$, and the concatenation of the extending words is a word reaching $S$. Consequently, starting with a set composed of a single state, if each extending word is not longer than $n$, it provides with a synchronizing word of length $(n-1)^2$.}, 
which was used to prove Conjecture~\ref{cernyconj} for particular classes of automata (see \cite{steinberg-2009} for a detailed analysis of this method).

In Section~\ref{Cex}, we study subsets reachability in synchronizing automata through the angle of Don's Conjecture \cite{don2016vcerny}. First, we show that the conjecture does not hold by providing a counterexample. Second, we analyze alternative versions of the conjecture and the questions they raise.

In Section~\ref{Gamma}, we study \emph{completely reachable automata}.
An automaton is said to be completely reachable if any of its proper subsets is reachable. Completely reachable automata are studied in depth in \cite{bondar2016completely}. 
A key tool for studying such automata is the $\Gamma_1-$graph derived from the automaton. In \cite{bondar2016completely}, Bondar and Volkov conjectured that this graph is strongly connected if some assumptions are fulfilled. We will show that this conjecture does not hold, but we will prove that it is true if the assumptions are slightly modified.

Due to the space constraints, the proof of the last proposition in this work is moved to the appendix.

\subsection*{Definitions}

A deterministic finite automaton (DFA) is a triple $A= (Q;\Sigma;\delta)$ with $Q$ a set of states, $\Sigma$ an alphabet of letters and $\delta: Q \times \Sigma \to Q$ a transition function.
It is convenient to represent DFA with a directed graph in which each state is represented by a node, and each transition $\delta(q_i, l_i)=q_j$ is represented by a directed edge from $q_i$ to $q_j$ labelled by letter $l_i$. For example, Fig.\ref{C4} shows an automaton with $Q=\{q_1, q_2, q_3, q_4\}$,  $\Sigma=\{a,b\}$, $\delta(q_1, a)=q_1$, $\delta(q_1, b)=q_2$, $\delta(q_2, a)=q_2$, $\delta(q_2, b)=q_3$, $\delta(q_3, a)=q_3$, $\delta(q_3, b)=q_4$, $\delta(q_4, a)=q_1$ and $\delta(q_4, b)=q_1$.  In the following, we will often directly use a graph to describe an automaton since we can recover the formal description $(Q;\Sigma;\delta)$ of the automaton from the graph representation and vice-versa. 

In this context, we define \emph{words} of $\Sigma^*$ as sequences of letters of $\Sigma$. We call \emph{rank} of a word $w$ the cardinality of its image $|Qw|$. Transition functions are recursively extended to words as follows: for a word $w=l_iw'$ in $\Sigma^*$, with $l_i$ a letter and $w'$ a word, $\delta(q_i, w)=\delta(\delta(q_i, l_i), w')$. Similarly, transition functions are extended to sets of states as follows: for a set $S\in Q$ and a word $w\in\Sigma^*$, $\delta(S, w)=\{q_j | q_j=\delta(q_i, w), q_i \in S\}$. In order to simplify the notation, we write $q_iw$ for $\delta(q_i, w)$ and $Sw$ for $\delta(S, w)$.

We define the \emph{power automaton} of an automaton $A=(Q;\Sigma;\delta)$ as the automaton obtained by setting all the possible subsets of $Q$ as states, maintaining the same letters as $A$ and the transition function equal to the transition function of $A$ on its sets of states. The \emph{square graph} of an automaton is equal to the restriction of the power automaton to pairs. A set $S\in Q$ is called \emph{reachable} if there exists a word $w\in \Sigma^*$ such that $Qw=S$. This is equivalent to having a path from $Q$ to $S$ in the power automaton. A set $S\in Q$ is called \emph{extendable} by a word $w$ if $S=S'w$ for $S'\subseteq Q$ with $|S'|>|S|$. The \emph{diameter} of a graph is the maximal length of the shortest path between two of its states. We say that states $q_i$ and $q_j$ define the diameter if the shortest path from $q_i$ to $q_j$ has a length equal to the diameter.



An automaton is \emph{sychronizing} if it has a reachable set of size 1. A word $w$ such that $|Qw|=1$ is said to be a \emph{synchronizing word}. 
For example, the automaton in Fig.\ref{C4} is synchronizing as the word $w=abbbabba$ is such that $Qw=q_1$, and therefore $w$ is a synchronizing word.






For each word $w\in \Sigma^*$ of rank $n-1$, we name the state $Q\backslash Qw$ as $excl(w)$, and the state $q$ such that there exists two states $q_1$ and $q_2$ with $q_1w=q_2w=q$ as $dupl(w)$.

For an automaton $\mathcal{A}= (Q;\Sigma;\delta)$, we define the directed graph $\Gamma_1(\mathcal{A})=(N,V)$, with node set $N=Q$ and edge set
$V=\{(q_i, q_j)| q_i=excl(w), q_j=dupl(w), \text{ for some w }\in \Sigma^*\text{ and }|Qw|=|Q|-1\}$.

For example, on the left of Fig.~\ref{Reachable} we present a completely reachable automaton in which every set $S\subset Q$ is reachable with words composed of subwords of rank $n-1$. Its corresponding $\Gamma_1-$graph is represented on the right of Fig.~\ref{Reachable}, with the edges labelled by a word inducing them, for the sake of clarity. We notice that, in this case, $\Gamma_1$ is strongly connected.

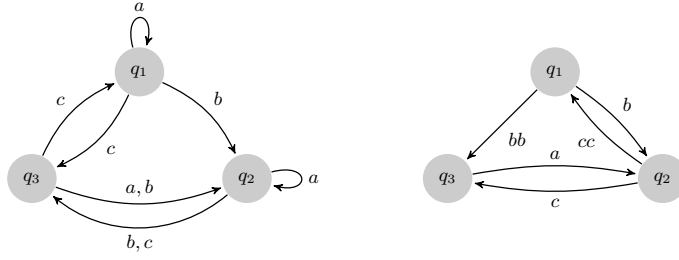
\begin{figure}[ht]
\begin{center}
\scalebox{0.80}{
\begin{tikzpicture}[->,>=stealth',shorten >=1pt,auto,node distance=2.5cm,
                    semithick]
  \tikzstyle{every state}=[fill=light-gray,draw=none,text=black, scale=1]

   \node[state] (A)                    {$q_{1}$};
  \node[state]         (B) [below right of=A] {$q_{2}$};
  \node[state]         (C) [ below left of =A] {$q_{3}$};
  
  \node[state] (A2)      [right=6cm of A]              {$q_{1}$};
  \node[state]         (B2) [below right of=A2] {$q_{2}$};
  \node[state]         (C2) [ below left of =A2] {$q_{3}$};

  \path (A) edge        [loop above]      node  {$a$} (B)
  			 edge      [bend left=20]       node {$b$} (B)
  			 edge      [bend left=20]       node {$c$} (C)
        (B) 
        	edge       [loop right]       node {$a$} (B)
        	edge         [bend left=40]    node  {$b,c$} (C)
        (C) edge     [bend left=20]         node {$ c$} (A)
        	edge     [bend right=20]         node {$a,b$} (B)
        (A2) 
  			 edge       [bend left=10]     node {$b$} (B2)
  			 edge            node {$bb$} (C2)
        (B2) 
        	edge         [bend left=10]    node  {$c$} (C2)
        	 edge       [bend left=10]     node {$cc$} (A2)
        (C2)
        	edge     [bend left=10]         node {$a$} (B2)
        ;

\end{tikzpicture}
}
\end{center}
\caption{A completely reachable automaton and the $\Gamma_1$ graph associated}
\label{Reachable}
\end{figure}



\section{Subset Reachability}\label{Cex}

In this section, we provide an answer to Don's conjecture on subset reachability and analyze questions arising from it.
\begin{conj}[Conjecture 2 in \cite{don2016vcerny}]
\label{Don}
Let $A= (Q;\Sigma;\delta)$ be an $n-$state automaton.  If $S\subset Q$ is a set of size $k$ and there exists a word $w$ such that $\delta(Q;w)=S$, then there exists a word with this property of length at most $n(n-k)$.
\end{conj}
This conjecture would imply that Conjecture~\ref{cernyconj} is also true \cite{don2016vcerny}.

In order to analyze the conjecture, we need the following lemma:

\begin{lemma}
\label{shortestword}
Let $\mathcal{A}= (Q;\Sigma;\delta)$ be an automaton. Let $S\subset Q$ be a reachable subset, and $w\in \Sigma^*$ be a shortest word such that $Qw=S$. Let us write $w=l_1\dots l_k$, with $k$ the length of $w$ and $l_i\in \Sigma$ the letters of $w$, with $0<i\leq k$. Let the subsets $S_i=Ql_1\dots l_i$ with $0<i<k$ be the subsets reached by prefixes of $w$. 

Then, all the subsets $S_i$ are different, and they all have a cardinality larger or equal to the cardinality of $S$. 
\end{lemma}
\begin{proof}
If two subsets $S_i=Ql_1\dots l_i$ and $S_j=Ql_1\dots l_j$, with $i<j$ are equal, then the word $l_1\dots l_i l_{j+1} \dots l_k$, which is shorter than $w$, also reaches $S$, which is a contradiction.

In a complete deterministic automaton, each state has exactly one image. Therefore for any word $w$ and and set $S$, we have $|Sw|\leq |S|$.
\end{proof}

This lemma implies that any reachable set of size $n-1$ can be reached with a word of length $n$. Indeed, there are only $n$ different subsets of cardinality higher or equal to the one of this set. Therefore, we will search for a counterexample to Conjecture~\ref{Don} with a set of size $n-2$ which cannot be reached with a word of length lower or equal to $2n$.
 
One way to build such automata is to select the letters in the following way. First choose a set of permutation letters such that the square graph of the automaton limited to these letters has a large diameter $D$. Second, add a letter $l$ of rank $n-2$ such that $Q\backslash Ql$ is the first pair defining the diameter of the square graph of the permutation letters. With this construction, the set containing all the states except the second pair defining the diameter cannot be reached by words shorter than $D+1$. 

In Fig.\ref{NewSet}, we present the automaton $\mathcal{P}_{2,n}$, with $n$ congruent to 3 modulo 4, which is built in that way. 

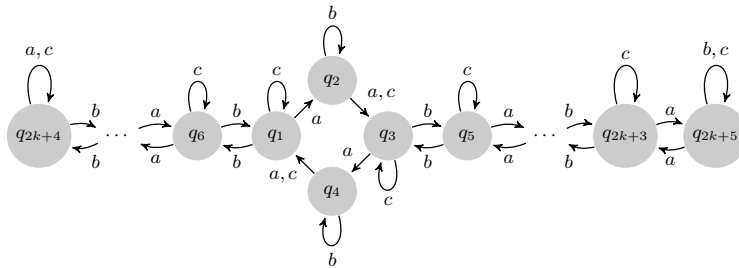
\begin{figure}[ht]
\begin{center}
\scalebox{0.80}{
\begin{tikzpicture}[->,>=stealth',shorten >=1pt,auto,node distance=1.3cm,
                    semithick]
  \tikzstyle{every state}=[fill=light-gray,draw=none,text=black, scale=1]

   \node[state] (A)                    {$q_{1}$};
  \node[state]         (B) [above right of=A] {$q_{2}$};
  \node[state]         (C) [ below right of =B] {$q_{3}$};
  \node[state]         (D) [below left of =C] {$q_{4}$};  \node[state]         (E) [right of =C] {$q_{5}$};
   \node[state]         (G) [left of =A] {$q_{6}$};

       \node         (K) [right of =E] {$\cdots$};
       \node         (L) [left of =G] {$\cdots$};

       \node[state]         (F) [right of =K] {$q_{2k+3}$};
       \node[state]         (H) [right=0.4cm of F] {$q_{2k+5}$};
       \node[state]         (J) [left of =L] {$q_{2k+4}$};

  \path (A) edge              node [swap] {$a$} (B)
  			 edge      [bend left=20]       node {$b$} (G)
  			 edge       [loop above]       node {$c$} (B)
        (B) edge             node {$a, c$} (C)
        	edge       [loop above]       node {$b$} (B)
        (C) edge             node [swap] {$a$} (D)
        	edge     [bend left=20]         node {$b$} (E)
        	edge       [loop below]       node {$c$} (B)
        (D) edge  	         node  {$a,c$} (A)
       		edge    [loop below]          node {$b$} (A)
			
        (E) edge    [bend left=20]       node {$b$} (C)
        edge     [bend left=20]     node {$a$} (K)
        edge       [loop above]       node {$c$} (B)
        (F)edge    [bend left=20]       node {$b$} (K)
        edge      [bend left=20]      node {$a$} (H)
        edge       [loop above]       node {$c$} (B)
        (G)edge    [bend left=20]       node {$b$} (A)
        edge      [bend left=20]      node {$a$} (L)
        edge       [loop above]       node {$c$} (B)
        (H)edge    [bend left=20]       node {$a$} (F)
        edge     [loop above]      node {$b,c$} (E)
        (K)edge    [bend left=20]       node {$b$} (F)
        edge     [bend left=20]      node {$a$} (E)
        
        (L)edge    [bend left=20]       node {$b$} (J)
        edge     [bend left=20]      node {$a$} (G)
        (J)edge    [bend left=20]       node {$b$} (L)
        edge       [loop above]       node {$a,c$} (B);

\end{tikzpicture}
}
\end{center}
\caption{The automaton $\mathcal{P}_{2,n}$}
\label{NewSet}
\end{figure} 

\begin{proposition}
\label{P2n}
The shortest word reaching the set $Q\backslash \{q_{k+2},q_{k+4}\}$ in the automaton in $\mathcal{P}_{2,n}$ is of length $n^2/4+5n/4-6$.
\end{proposition}
\begin{Long}
\begin{proof}
We first notice that letters $a$ and $b$ of the automaton $\mathcal{P}_{2,n}$ are a set of permutation letters studied in \cite{GonzeGGJV17} such that the shortest path from the pair $q_2q_4$ to $q_{k+2}q_{k+4}$ is of length $n^2/4+5n/4-7$ (see \cite{GonzeGGJV17}). As $a$ and $b$ are permutations, we have $Qa=Q$ and $Qb=Q$, so the first letter of any shortest word reaching a set $S$ should be $c$.

Second, we notice that letter $c$ sends $q_2$ to $q_3$ and $q_4$ to $q_1$, and the other states to themselves, so $Qc=Q\setminus\{q_2, q_4\}$ and is of cardinality $n-2$. Moreover, if letter $c$ is applied to any set containing $n-2$ states, it either sends it to $Q\setminus\{q_2, q_4\}$, or to a set of lower cardinality. Therefore, due to Lemma~\ref{shortestword}, $c$ can only be at the first position in a shortest word reaching $Q\setminus\{q_2, q_4\} $.

Third, applying letters $a$ and $b$ to a set containing $n-2$ states is equivalent to applying it to the complementary set, i.e. to pairs. Therefore, the shortest word $w$ mapping the set $Q\setminus\{q_2, q_4\} $ to $Q\setminus\{q_{k+2}, q_{k+4}\} $ is also the shortest word mapping $q_2q_4$ to $q_{k+2}q_{k+4}$, which is of length $n^2/4+5n/4-7$. Therefore, $cw$ is the shortest word reaching $Q\setminus\{q_{k+2}, q_{k+4}\} $, and is of length $n^2/4+5n/4-6$. 
\end{proof}
\end{Long}

The automaton of $\mathcal{P}_{2,n}$ is a generic counterexample to Conjecture~\ref{Don}, which claims that the distance from $Q$ to any set of size $n-2$ should be $2n$.

We notice that the reachable sets of $\mathcal{P}_{2,n}$ can be reached with words of polynomial length. Therefore, one can wonder if it is a general property:  
\begin{problem}
\label{PolConjConj}
Let $A= (Q;\Sigma;\delta)$ be an $n-$state automaton.  If $S\in Q$ is a set of size $k$ and there exists a word $w$ such that $Qw =S$, does there exist a word with this property of polynomial length with respect to $k$?
\end{problem}

\subsubsection{Answer to Problem~\ref{PolConjConj}: Exponential Length Shortest Words}
For any $n$, we can build an automaton such that it has a set $S\in Q$ containing $\lfloor \frac{n}{2}\rfloor$ states which cannot be reached with a word shorter than $L=C^{n-1}_{\lfloor \frac{n}{2}\rfloor}$, which is an exponential quantity with respect to $n$. This will be shown in Definition~\ref{P3} and Proposition~\ref{P3n} below.

\begin{definition}
\label{P3}
We define the automaton $\mathcal{P}_{3,n}$ as follows. 
It has $n$ states $q_1\dots q_n$, a letter $a$ and letters $l_i$ that we will define later. The letter $a$ is defined as follows:

$q_1 a=q_2$

$q_i a=q_i$ for $1<i\leq \lfloor \frac{n}{2}\rfloor+1$

$q_i a=q_2$ for $\lfloor \frac{n}{2}\rfloor+1< i \leq n.$

In order to define the other letters, we first list all the sets composed of $\lfloor\frac{n}{2}\rfloor$ states among states $q_2, \dots, q_n$. Namely we take $S_1$ the set composed of $q_2, \dots, q_{\lfloor \frac{n}{2}\rfloor+1}$ and we list all the other possible sets as $ S_2, \dots, S_{L}$, with $S_{L}=\{q_{n-\lfloor \frac{n}{2}\rfloor+1},\dots,q_n\}$. We notice that there are exactly $L= C^{n-1}_{\lfloor \frac{n}{2}\rfloor}$ such sets. In each set $S_i$, we order the elements according to their indices.

We now define letters $l_1, \dots, l_{L-1}$. First, we have $q_1 l_i=q_1$ for $1\leq i\leq L-1$.
Then, the effect of $l_i$ on the other states is defined as follows: $ q_j l_i=q_1$ if $q_j \notin S_i$ and $ q_j l_i=q_k$ if $q_j \notin S_i$, with $q_j$ the $m$th element of $S_i$ and $q_k$ the $m$th element of $S_{i+1}$. Therefore, each of the letters $l_i$ injects the elements of $S_i$ in $S_{i+1}$ and sends all the elements not in $S_i$ to $q_1$.
\end{definition}

In Fig.\ref{ExpCount}, we show a representation of $\mathcal{P}_{3,n}$ (we do not represent the effect of $l_1, \dots, l_L$ on other states than $q_1$ for the sake of clarity).

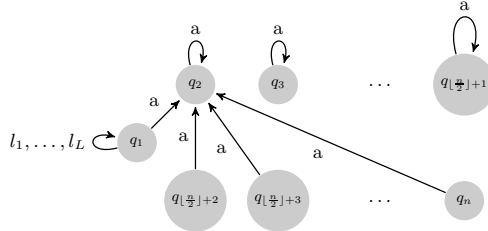
\begin{figure}[h!]
\begin{center}
\scalebox{0.8}{
\begin{tikzpicture}[->,>=stealth',shorten >=1pt,auto,node distance=1.7cm,
                    semithick]
  \tikzstyle{every state}=[fill=light-gray,draw=none,text=black, scale=0.8]

  \node[state] (A)                    {$q_1$};
  \node[state]         (B) [above right of=A] {$q_2$};
  \node[state]         (C) [right  of=B] {$q_3$};
  \node         (D) [right  of=C] {$\dots$};
  \node  [state]       (E) [right  of=D] {$q_{\lfloor \frac{n}{2}\rfloor+1}$};
  \node[state]         (F) [below right of=A] {$q_{\lfloor \frac{n}{2}\rfloor+2}$};
  \node[state]         (G) [ right of=F] {$q_{\lfloor \frac{n}{2}\rfloor+3}$};
  \node         (H) [right  of=G] {$\dots$};
  \node  [state]       (I) [right  of=H] {$q_{n}$};

  \path (A) edge              node {a} (B)
  			edge      [loop left]        node {$l_1, \dots, l_L$} (B)
   (B) edge    [loop above]          node {a} (B)
  (C) edge      [loop above]        node {a} (B)
 (E) edge      [loop above]        node {a} (B)
  (F) edge              node {a} (B)
   (G) edge              node {a} (B)
   (I) edge              node {a} (B)
            ;
\end{tikzpicture}
}
\end{center}
\caption{Part of the automaton $\mathcal{P}_{3,n}$}
\label{ExpCount}
\end{figure}

\begin{proposition}
The shortest word reaching the set $\{q_{n-\lfloor \frac{n}{2}\rfloor+1},\dots,q_n\}$ in the automaton $\mathcal{P}_{3,n}$ is $C^{n-1}_{\lfloor \frac{n}{2}\rfloor}$ letters long.
\label{P3n}
\end{proposition}
\begin{Long}
\begin{proof}

From the definition of letters $l_i$, it is clear that the word $al_1l_2\dots l_{L-1}$ maps $Q$ on $S_L$. 

Moreover, it is the shortest of such words. Indeed, let $w$ be a word such that $Qw=S_L$. We first notice that the only letter which maps $q_1$ on another state is $a$. Therefore, as $Qw$ does not include $q_1$, it must contain a last letter $a$ and letters after it which do not map any state of the transition sets on $q_1$. 
Now, if we consider a set $S_i$ of $\lfloor \frac{n}{2}\rfloor$ states from $q_2, \dots, q_n$, the only letter $l$ which does not map any states of $S_i$ on $q_1$ is $l_i$. 
Therefore, if a subset $S_i$ was reached with a prefix of $w$ containing the last $a$ of $w$, the only possible letter $l$ such that $q_1\notin S_il$ is $l_i$. This implies that the letter following the prefix is $l_i$ and that the next subset is $S_il=S_{i+1}$. Since $Qa=\{q_2, \dots, q_{\lfloor \frac{n}{2}\rfloor+1}\}$ is already of cardinality $\lfloor \frac{n}{2}\rfloor$, the first subset after the last $a$ must be $\{q_2, \dots, q_{\lfloor \frac{n}{2}\rfloor+1}\}=S_1$. Then, by induction, the letters following the last $a$ in $w$ are $l_1, \dots, l_{L-1}$, in that order. Therefore, $w$ has a suffix $al_1l_2\dots l_{L-1}$, which is $C^{n-1}_{\lfloor \frac{n}{2}\rfloor}$ letters long. Since this suffix is a word reaching $S_L$, it is also the shortest word reaching $S_L$.
\end{proof}
\end{Long}


We notice that the families of automata presented were such that a first letter was used to lower the cardinality of the set, and then only permutations were used. Applying the first letter again would have reduced the number of states in the set below the cardinality of $S$ or reach a set already reached, which is forbidden by Lemma~\ref{shortestword}.
Moreover, it turns out that these automata have a set $S$ which cannot be reached with a short word, but such that some subset of $S$ can be reached with such word. This leads to the following question:
\begin{problem}
\label{WeakerConj}
Let $A= (Q;\Sigma;\delta)$ be an $n-$state automaton.  If $S\in Q$ is a set of size $k$ and there exists a word $w$ such that $Qw\subseteq S$, does there exist a word with this property of length at most $n(n-k)$?
\end{problem}

\subsubsection{Answer to Problem~\ref{WeakerConj}: Included subsets}
We notice that, although this question is weaker than Conjecture~\ref{Don}, a positive answer would still imply that Cerny's conjecture is true.

However, even for this weaker version, the answer is negative. Indeed, the automaton $\mathcal{P}_4$ in Fig.~\ref{EmptyState}, which was introduced in \cite{GJT2014}, is such that the set $S=\{1, 2 ,3\}$ or any subset of $S$ cannot be reached with a word of length lower than 6. This can be verified by breadth first search algorithm on the power automaton, presented in Fig.\ref{EmptyStatePow}.

\begin{figure}[h!]
\begin{center}
\scalebox{0.8}{
\begin{tikzpicture}[->,>=stealth',shorten >=1pt,auto,node distance=1.7cm,
                    semithick]
  \tikzstyle{every state}=[fill=light-gray,draw=none,text=black, scale=0.9]

  \node[state] (A)                    {$0$};
  \node[state]         (B) [above right of=A] {$1$};
  \node[state]         (D) [below right of=A] {$3$};
  \node[state]         (C) [below right of=B] {$2$};

  \path (A) edge              	node {a} (B)
            edge [loop left]  	node {b} (C)
        (B) edge   				node {} (A)
            edge              	node {b} (C)
        (C) edge              	node {b} (D)
         	edge [loop right] 	node {a} (D)
        (D) edge  			  	node {b} (B)
        	edge  			  	node {a} (A);
\end{tikzpicture}}
\end{center}
\caption{The automaton $\mathcal{P}_4$}
\label{EmptyState}
\end{figure}
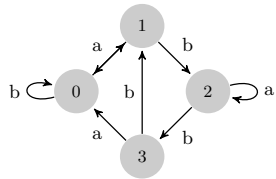

\begin{figure}[h!]
\begin{center}
\scalebox{0.8}{
\begin{tikzpicture}[->,>=stealth',shorten >=1pt,auto,node distance=1.7cm,
                    semithick]
  \tikzstyle{every state}=[fill=light-gray,draw=none,text=black, scale=0.9]

  \node[state] (A)                    {$0$};
  \node[state]         (B) [above right of=A] {$1$};
  \node[state]         (D) [below right of=A] {$3$};
  \node[state]         (C) [below right of=B] {$2$};
  \node[state] (E)          [below left of=A]          {$13$};
  \node[state] (F)          [below  of=E]          {$23$};
  \node[state] (G)          [left  of=F]          {$12$};
  \node[state] (H)          [above  of=G]          {$03$};
  \node[state] (I)          [left  of=H]          {$02$};
  \node[state] (J)          [above  of=I]          {$01$};
  \node[state] (K)          [left  of=J]          {$013$};
  \node[state] (L)          [below  of=K]          {$023$};
  \node[state] (M)          [left  of=L]          {$012$};
  \node[state] (N)          [above  of=M]          {$0123$};
  \node[state] (O)          [below of =L]          {$123$};

  \path (A) edge              	node {a} (B)
            edge [loop left]  	node {b} (C)
        (B) edge   				node {} (A)
            edge              	node {b} (C)
        (C) edge              	node {b} (D)
         	edge [loop right] 	node {a} (D)
        (D) edge  			  	node {b} (B)
        	edge  			  	node {a} (A)
        (E) edge  			  	node {b} (G)
        	edge  			  	node {a} (A)
        (F) edge  			  	node {b} (E)
        	 edge  		[bend left=70]	  	node {a} (I)
        (G) edge  			  	node {b} (F)
        	edge  		[bend left=10]	  	node {a} (I)
        (H) edge  			  	node {a,b} (E)
        (I) edge  			  	node {b} (H)
        	edge  		[bend left=10]	  	node {a} (G)
        (J) edge  	[loop right]		  	node {a} (B)
        	edge  			  	node {b} (I)
        (K) edge  			  	node {b} (M)
        	edge  			  	node {a} (J)
        (L) edge  			  	node {b} (K)
        	edge  	[bend left=10]		  	node {a} (M)
        (M) edge  [bend left=10]			  	node {b} (L)
        	edge [loop left]  	node {a} (C)
        (N) edge  			  	node {a} (M)
        	edge [loop left]  	node {b} (C)
        (O) edge  			  	node {a} (I)
        	edge [loop left]  	node {b} (C);
\end{tikzpicture}}
\end{center}
\caption{Power automaton of $\mathcal{P}_4$}
\label{EmptyStatePow}
\end{figure}
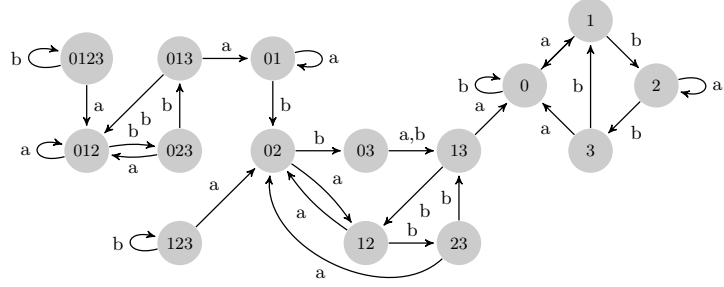

\section{$\Gamma_1$-graph of completely reachable automata}\label{Gamma}


It is noticed by Don in \cite{don2016vcerny} that if the graph $\Gamma_1(\mathcal{A})$ is cyclic, then the automaton is completely reachable. In \cite{bondar2016completely}, it is noticed that in fact, complete reachability holds under the weaker assumption of strong connectivity. However, the converse is not true: there are completely reachable automata for which $\Gamma_1(\mathcal{A})$ is not strongly connected. The second problem that we consider in this paper is the conjecture presented by Bondar and Volkov in \cite{bondar2016completely}:
\begin{conj}[Conjecture 3 in \cite{bondar2016completely}]
\label{Bondar}
If for every proper non-empty subset $P$ of the state set of a DFA $\mathcal{A}= (Q;\Sigma;\delta)$ there is a product $w$ of words of rank $|Q|-1$ with respect to $\mathcal{A}$ such that $P=Qw$, the graph $\Gamma_1(\mathcal{A})$ is strongly connected.
\end{conj}
We provide an analysis of the way letters of rank $n-1$ combine with each other, and their influence on $\Gamma_1(\mathcal{A})$. This analysis allows to build a counterexample to Conjecture \ref{Bondar} and to prove a modified version of Conjecture~\ref{Bondar}.

In order to formalize the link between strong connectivity of the graph $\Gamma_1$ and subset reachability, we need the following definition of a strongly connected graph:

\begin{definition}
A directed graph $\Gamma=(N,E)$ with nodes $N$ and edges $E$ is strongly connected iff for any set $S\subset N$, there exists an edge $e=(n_1, n_2)\in E$ with $n_1 \notin S$ and $n_2\in S$. 
\end{definition}

We say that the edge $e=(n_1, n_2)\in E$ with $n_1 \notin S$ and $n_2\in S$ is \emph{intersecting} the set $S$.

This definition implies that iff $\Gamma_1$ is strongly connected, then any set $S\subset Q$ has an intersecting edge in $\Gamma_1$. This means that there exists a word $w\in\Sigma^*$ of rank $n-1$ with $dupl(w)\in S$ end $excl(w)\notin S$, and therefore that $w$ is extending the set $S$.




In order to analyze the effect of word concatenation, we cannot restrict ourselves to the concepts of $excl$ and $dupl$ of a word. Indeed, for two letters $a$ and $b$, the knowledge of $excl$ and $dupl$ is not enough to determine whether $ab$ or $ba$ are of rank $n-1$ or of rank $n-2$. To analyze this question, we need an additional concept: the \emph{roots} of a word of rank $n-1$. The roots are the states which are sent on the duplicated state. More formally, for a word $w$ of rank $n-1$, $root(w)=\{q_i|q_iw=dupl(w)\}$. We then have the following property:
\begin{lemma}
\label{key}
Let $A= (Q;\Sigma;\delta)$ be a DFA. Let $w_1, w_2 \in \Sigma^*$ be two words of rank $|Q|-1$. If $excl(w_1)\notin root(w_2)$, then $w_1w_2$ is of rank $n-2$. If $excl(w_1)\in root(w_2)$ then $w_1w_2$ is of rank $|Q|-1$. In that case, $excl(w_1w_2)=excl(w_2)$, $dupl(w_1w_2)=dupl(w_1)w_2$ and $root(w_1w_2)=root(w_1)$. 
\end{lemma}
\begin{Long}
\begin{proof}
If $emp(w_1)\notin root(w_2)$, then $root(w_2)\subset Qw_1$, and the two states in $root(w_2)$ are synchronized by $w_1$ while the other states are permuted, so we have $|Qw_1w_2|=|Qw_1|-1=|Q|-2$.

Moreover, if $emp(w_1)\in root(w_2)$, then all the states in $Qw_1$ have a different image with $w_2$ and $|Qw_1w_2|=n-1$. In that case, we have the following properties.

The state $dupl(w_1w_2)$ is the image of $dupl(w_1)$ by word $w_2$, i.e. $dupl(w_1)w_2$.

The states $root(w_1w_2)$ are equal to $root(w_1)$, since for these states we have that $q_rw_1\in dupl(w_1)$ and therefore $q_rw_1w_2\in dupl(w_1)w_2$. 

The state $excl(w_1w_2)$ is $excl(w_2)$, since $excl(w_2)\notin Qw_2$ and $Qw_1w_2\subseteq Qw_2$.

\end{proof}
\end{Long}

This key lemma allows us to build the $\Gamma_1$-graph algorithmically.  Indeed, we can identify which words can be concatenated into other words of rank $n-1$, which induce edges in $\Gamma_1$. If two words combine into a word of length $n-2$, then the concatenation does not contribute to $\Gamma_1$.
 
We now define in Fig.~\ref{Lettersef} the automaton $\mathcal{P}_4$, which has six states and six letters of rank 5. Using Lemma \ref{key}, we will prove that it is a counterexample to Conjecture \ref{Bondar}.

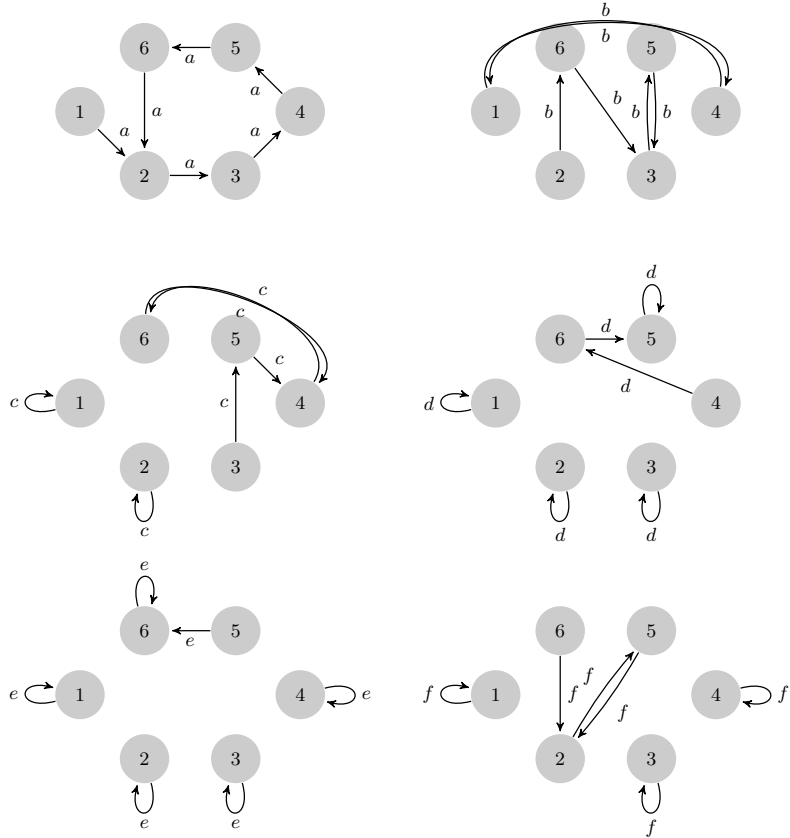
\begin{figure}[ht]
\begin{center}
\scalebox{0.80}{
\begin{tikzpicture}[->,>=stealth',shorten >=1pt,auto,node distance=1.5cm,
                    semithick]
  \tikzstyle{every state}=[fill=light-gray,draw=none,text=black, scale=1]

   \node[state] (A)                    {$1$};
  \node[state]         (B) [below right of=A] {$2$};
  \node[state]         (C) [ right of =B] {$3$};
  \node[state]         (D) [ above right of =C] {$4$};
  \node[state]         (E) [ above left of =D] {$5$};
  \node[state]         (F) [ left of =E] {$6$};
  
  \node[state] (A2)      [right=6cm of A]              {$1$};
  \node[state]         (B2) [below right of=A2] {$2$};
  \node[state]         (C2) [ right of =B2] {$3$};
  \node[state]         (D2) [ above right of =C2] {$4$};
  \node[state]         (E2) [ above left of =D2] {$5$};
  \node[state]         (F2) [left of =E2] {$6$};
  
  \node[state] (A3)      [below=4cm of A]              {$1$};
  \node[state]         (B3) [below right of=A3] {$2$};
  \node[state]         (C3) [ right of =B3] {$3$};
  \node[state]         (D3) [ above right of =C3] {$4$};
  \node[state]         (E3) [ above left of =D3] {$5$};
  \node[state]         (F3) [left of =E3] {$6$};
  
  \node[state] (A4)      [right=6cm of A3]              {$1$};
  \node[state]         (B4) [below right of=A4] {$2$};
  \node[state]         (C4) [ right of =B4] {$3$};
  \node[state]         (D4) [ above right of =C4] {$4$};
  \node[state]         (E4) [ above left of =D4] {$5$};
  \node[state]         (F4) [left of =E4] {$6$};
  
     \node[state] (A5)      [below=4cm of A3]                {$1$};
  \node[state]         (B5) [below right of=A5] {$2$};
  \node[state]         (C5) [ right of =B5] {$3$};
  \node[state]         (D5) [ above right of =C5] {$4$};
  \node[state]         (E5) [ above left of =D5] {$5$};
  \node[state]         (F5) [ left of =E5] {$6$};
  
  \node[state] (A6)      [right=6cm of A5]              {$1$};
  \node[state]         (B6) [below right of=A6] {$2$};
  \node[state]         (C6) [ right of =B6] {$3$};
  \node[state]         (D6) [ above right of =C6] {$4$};
  \node[state]         (E6) [ above left of =D6] {$5$};
  \node[state]         (F6) [left of =E6] {$6$};

  \path (A) edge             node  {$a$} (B)
        (B) edge              node  {$a$} (C)
        (C) edge           node  {$a$} (D)
        (D) edge            node  {$a$} (E)
        (E) edge             node  {$a$} (F)
        (F) edge            node  {$a$} (B)
        
        (A2) edge        [bend left=110]     node  {$b$} (D2)
        (B2) edge              node  {$b$} (F2)
        (C2) edge       [bend left=5]      node  {$b$} (E2)
        (D2) edge       [bend left=-100]      node  {$b$} (A2)
        (E2) edge       [bend left=5]     node  {$b$} (C2)
        (F2) edge             node  {$b$} (C2)
        
        (A3) edge        [loop left]     node  {$c$} (D3)
        (B3) edge       [loop below]       node  {$c$} (F3)
        (C3) edge             node  {$c$} (E3)
        (D3) edge       [bend left=-100]      node  {$c$} (F3)
        (E3) edge           node  {$c$} (D3)
        (F3) edge         [bend left=110]    node  {$c$} (D3)
        
        (A4) edge        [loop left]     node  {$d$} (D4)
        (B4) edge       [loop below]       node  {$d$} (F4)
        (C4) edge       [loop below]     node  {$d$} (E4)
        (D4) edge            node  {$d$} (F4)
        (E4) edge       [loop above]     node  {$d$} (E4)
        (F4) edge             node  {$d$} (E4)
        
        (A5) edge        [loop left]      node  {$e$} (B5)
        (B5) edge        [loop below]      node  {$e$} (B5)
        (C5) edge       [loop below]      node  {$e$} (B5)
        (D5) edge       [loop right]      node  {$e$} (B5)
        (E5) edge             node  {$e$} (F5)
        (F5) edge       [loop above]      node  {$e$} (B5)
        
        (A6) edge        [loop left]      node  {$f$} (B)
        (B6) edge        [bend left=5]      node  {$f$} (E6)
        (C6) edge       [loop below]      node  {$f$} (B)
        (D6) edge       [loop right]      node  {$f$} (B)
        (E6) edge       [bend left=5]     node  {$f$} (B6)
        (F6) edge             node  {$f$} (B6)
        ;
        ;

\end{tikzpicture}
}
\end{center}
\caption{The $6$ letters of the automaton $\mathcal{P}_4$}
\label{Lettersef}
\end{figure}

%


\begin{proposition}
The automaton $\mathcal{P}_4$ is completely reachable, and the graph $\Gamma_1(\mathcal{P}_4)$ is not strongly connected.
\label{P4}
\end{proposition}
\begin{Long}
\begin{proof}
Figure~\ref{TableWord} lists all the words of rank 5 of this automaton, with their duplicated states, excluded states and root states. In the first six lines are listed the single letters. We notice that, due to Lemma \ref{key}, the only words or rank $n-1$ are individual letters, $a^*$, $\{e,f\}^*$, $\{f,g\}^*fa^*$ and $\{f,g\}^*l$ with $l$ being any of the letters. Indeed, the roots of any letter are states $1$, $5$ and $6$, and the only letters with $1$, $5$ or $6$ as excluded states are $e$ and $f$. Therefore, in order to keep words of rank 5, letter $a$ can only be preceded by letters $a$ or $f$, and all the other letters, including $e$ and $f$, can only be preceded by letters $e$ or $f$, which leads to the six last lines of the table. The duplicated states are found by applying the combination rule of Lemma~\ref{key}.


The edges of the graph $\Gamma_1$ are defined by the empty states and double states of words in Fig.\ref{TableWord}, i.e. the second and third column. Since state 1 does not appear in the duplicated state column, there are no edge of $\Gamma_1$ directed toward it, and therefore the graph $\Gamma_1$ is not strongly connected.

We notice however that any set is reachable. The edges of $\Gamma_1$ obtained with single letters form a cycle with states 2, 3, 4, 5 and 6, plus an edge from 1 to 2. In this setting, any subset of $Q$ except for set $\{1\}$ has an intersecting edge. Therefore, all sets of states except for $\{1\}$ are extendable, which implies by induction on the number of states that they are also reachable. The set $\{1\}$ itself can also be reached from set $\{4\}$ by applying letter $b$. Therefore, the automaton is completely reachable and has a $\Gamma_1$ graph which is not strongly connected. So it is a counterexample to Conjecture \ref{Bondar}.
\end{proof}
\end{Long}

\begin{figure}[ht]
\begin{center}
\begin{tabular}{|c|c|c|c|}
   \hline
   word & excluded state & duplicated state & root states \\
   \hline
   $a$ & 1 & 2 & 1, 6 \\
   $b$ & 2 & 3 & 5, 6 \\
   $c$ & 3 & 4 & 5, 6 \\
   $d$ & 4 & 5 & 5, 6 \\
   $e$ & 5 & 6 & 5, 6 \\
   $f$ & 6 & 2 & 5, 6 \\
   $a^*$ & 1 & 2, 3, 4, 5 or 6 & 1, 6 \\
   $\{e,f\}^*$ & 5 or 6 & 2, 5 or 6 & 5, 6 \\
   $\{e,f\}^*fa^*$ & 1 & 2, 3, 4, 5 or 6  & 5, 6 \\
   $\{e,f\}^*b$ & 2 & 2b=6, 5b=3 or 6b=3 & 5, 6 \\
   $\{e,f\}^*c$ & 3 & 2c=2, 5c=4 or 6c=4 & 5, 6 \\
   $\{e,f\}^*d$ & 4 & 2d=2, 5d=5 or 6d=5 & 5, 6 \\
   \hline
\end{tabular}
\caption{Exhaustive list of words of rank 5}
\label{TableWord}
\end{center}
\end{figure}

Conjecture \ref{Bondar} turns out to be false. However, we notice that, to reach the set $\{1\}$, the last letter we had to use was of rank $n-1$ while permuting the states of the set. This leads to the following observation: if a set $S$ is such that $Qw=S$, and such that $w=w_1w_2$, with $w_2$ a word of rank $n-1$, two possible cases arise. Either $|Qw_1|=|Qw_1w_2|$, or $|Qw_1|=|Qw|-1$. In the first case, the word $w_2$ acts as a permutation on the set $Qw_1$. In the latter case, $w_2$ has synchronized two states, and we observe the following:

\begin{lemma}
\label{LemmaIntersect}
Let $\mathcal{A}=\{Q, \Sigma, \delta\}$ be a DFA. For $S\subset Q$, if we have $Qw=S$, with $w=w_1w_2$, $w_2$ of rank $|Q|-1$ and $|Qw_1|=|Qw|+1$, then there is an edge intersecting $S$ in $\Gamma_1$.
\end{lemma}
\begin{Long}
\begin{proof}
The edge $(excl(w_2), dupl(w_2))$ is an edge of $\Gamma_1$ and intersects $S$.

First, $excl(w_2)\notin S$ because $S= Qw_1w_2\subseteq  Qw_2$, and $excl(w_2)\notin Qw_2$.

Second, we have $root(w_2)\subset Qw_1$, because otherwise $|Qw_1|=|Qw_1w_2|=|Qw|$, which contradicts the assumption that $|Qw_1|=|Qw|+1$. As $dupl(w_2)$ is the image of $root(w_2)$ by $w_2$, it implies that $dupl(w_2)$ is in $Qw_1w_2=S$. 

\end{proof}
\end{Long}

Based on this observation, we propose the following modified version of Conjecture \ref{Bondar}:

\begin{theorem}

\label{GonzeThm}
Let $\mathcal{A}=\{Q, \Sigma, \delta\}$ be a DFA. If for every proper non-empty subset $S\subset Q$ there is a product $w\in \Sigma^*$ such that $S=Qw$ and such that $w$ has a suffix $w_{end}$ of rank $|Q|-1$ with $|Qw|<|Q(w/w_{end})|$, then the graph $\Gamma_1(\mathcal{A})$ is strongly connected.
\end{theorem}
\begin{Long}
\begin{proof}
It is a direct application of Lemma~\ref{LemmaIntersect} on any set of states $S\subset Q$ of the automaton. Indeed, it implies that any set has an incoming edge in $\Gamma_1$, which is the definition of a strongly connected graph.
\end{proof}
\end{Long}

In particular, this theorem applies for automata with a simple idempotent, as studied in \cite{rystsov2000estimation}, and if the automaton generates the full transformation group, as studied in \cite{GonzeGGJV17}.

\section{Conclusion}



In this article, we focused on subset reachability in synchronizing automata. The first problem considered was the conjecture presented by Henk Don in \cite{don2016vcerny}.
We started by presenting a family of automata built from a set of permutation with a large square graph diameter which are counterexamples to Conjecture~\ref{Don}.
Then, we analyzed modified versions of Conjecture~\ref{Don} and natural questions which arise from it. In particular, we built a family of strongly connected synchronizing automata with subsets which cannot be reached with words shorter than a function $\Omega(exp(n))$.

The second problem that we considered is Conjecture~\ref{Bondar}, presented by Bondar and Volkov in \cite{bondar2016completely}.
We provided an analysis of the way letters of rank $n-1$ combine with each other, and their influence on the $\Gamma_1-$graph. This analysis allowed to build a counterexample to Conjecture \ref{Bondar} and to prove a modified version of Conjecture~\ref{Bondar}.

Finally, we propose the following problem, which is the restriction of Conjecture~\ref{Don} to completely reachable automata.

\begin{problem}
\label{CRArestriction}
Let $A= (Q;\Sigma;\delta)$ be an $n-$state completely reachable automaton. For any $0<k<n$ and any set $S\in Q$ of size $k$, is it true that there exists a word $w$ such that $Qw=S$ and $|w|\leq n(n-k)$?
\end{problem}
A positive answer to this problem would prove Cerny's conjecture for completely reachable automata.

\section*{Acknowledgements}

The authors would like to thank Balazs Gerencser and Vladimir Gusev for fruitful discussions and advice, and Elodie Boucquey, Myriam Gonze and Xavier Gonze for careful reading of the paper.

\bibliography{references}

\end{document}